\documentclass[a4paper,USenglish,cleveref, autoref,11pt]{article}

\usepackage[dvipsnames,usenames]{xcolor}

\usepackage{amsthm,amssymb,amsmath}  
\usepackage{xspace,enumerate}
\usepackage[utf8]{inputenc}
\usepackage{thmtools}
\usepackage{thm-restate}
\usepackage{authblk}
\usepackage{geometry}
\usepackage{lineno}
\RequirePackage{hyperref}
\hypersetup{colorlinks=true, urlcolor=Blue, citecolor=Green, linkcolor=BrickRed, breaklinks, unicode}

\usepackage{tikz}
\usetikzlibrary{arrows.meta}
 
\title{Fault-Tolerant Distance Labeling for Planar Graphs}
\author[1]{Aviv Bar-Natan}
\author[2]{Panagiotis Charalampopoulos}
\author[3]{Pawe\l{} Gawrychowski}
\author[2]{Shay Mozes}
\author[1]{Oren Weimann}

\affil[1]{
University of Haifa, Israel\\
\href{mailto:1aviv234@gmail.com}{1aviv234@gmail.com}, \href{mailto:oren@cs.haifa.ac.il}{oren@cs.haifa.ac.il}}
\affil[2]{
The Interdisciplinary Center Herzliya, Israel\\
\href{mailto:panagiotis.charalampopoulos@post.idc.ac.il}{panagiotis.charalampopoulos@post.idc.ac.il}, \href{mailto:smozes@idc.ac.il@gmail.com}{smozes@idc.ac.il}}
\affil[3]{
University of Wroc\l{}aw, Poland\\
\href{mailto:gawry@cs.uni.wroc.pl}{gawry@cs.uni.wroc.pl}}

\date{\vspace{-5ex}}

\usepackage{etoolbox}
\makeatletter
\setbool{@fleqn}{false}
\makeatother

\usepackage[ruled]{algorithm}
\usepackage[noend]{algpseudocode}
\algnewcommand{\LineComment}[1]{\State \(\triangleright\) #1}

\usepackage{graphicx}
\usepackage[noadjust]{cite} 
  \usepackage{microtype}
    \usepackage{xspace}
    \usepackage{amsmath,amsfonts}
    \usepackage{wasysym}
  \usepackage{makecell}
  \usepackage{tabularx}
  \usepackage{comment}
  \usepackage{todonotes}
  \usepackage{thmtools}
  \usepackage{thm-restate}
  \usepackage[capitalise]{cleveref}
  \usepackage{verbatim}
\usepackage{units}
\usepackage{color}
\definecolor{darkblue}{rgb}{0,0.08,0.45}
  \usepackage{epsfig}
    \usepackage{graphics}
\usepackage{pgf,tikz}
  \usetikzlibrary{trees, arrows, shapes, decorations}
\usepackage{marvosym}
 \usepackage[caption=false]{subfig}
\usetikzlibrary{decorations.pathmorphing}
\usetikzlibrary{decorations.markings}
\usetikzlibrary{decorations.pathmorphing,shapes}
\usetikzlibrary{calc,decorations.pathmorphing,shapes}
\usepackage{booktabs}
\usepackage{etoolbox}
\usepackage[title]{appendix}
\usepackage{enumitem}

\newtheorem{theorem}{Theorem}
\newtheorem{lemma}[theorem]{Lemma}
\newtheorem{conjecture}[theorem]{Conjecture}

\theoremstyle{remark}
\newtheorem*{remark}{Remark}

\newcommand{\cO}{O}
\newcommand{\TG}{\mathcal{T}}

\newcommand{\cOtilde}{\tilde{O}}

\begin{document}

\maketitle

\begin{abstract}
\normalsize
In fault-tolerant distance labeling we wish to assign short labels to the vertices of a graph $G$ such that from the labels of any three vertices $u,v,f$ we can infer the $u$-to-$v$ distance in the graph $G\setminus \{f\}$. We show that any directed weighted planar graph (and in fact any graph in a graph family with $\cO(\sqrt{n})$-size separators, such as minor-free graphs) admits fault-tolerant distance labels of size $\cO(n^{2/3})$. 
We extend these labels in a way that allows us to also count the number of shortest paths, and provide additional upper and lower bounds for labels and oracles for counting shortest paths. \end{abstract}

\section{Introduction}
Computing distances in graphs is one of the most basic and important problems in graphs theory, both from theoretical and practical points of view. In this work we consider distance labeling schemes, in which one preprocesses a network to assign labels to the vertices, so that the distance between any two vertices $u$ and $v$ can be recovered from just the labels of $u$ and $v$ (and no other information). The main criteria of interest are foremost the size of the label, and to a lesser extent the time it takes to recover the distance from a given pair of labels (query time). Distance labeling schemes are useful in the distributed setting, where it is advantageous to be able to infer distances based only on local information such as the labels of the source and destination. This is the case in communication networks or in disaster stricken areas, where communication with a centralized entity is infeasible or downright impossible.

Considering the latter scenario of disaster management, it is not only likely that a disastrous event makes communication with a centralized entity impossible, but also that parts of the network are affected by the disaster, and that only shortest paths that avoid affected parts of the network should be considered when computing distances. Forbidden-set distance labeling schemes assign labels to vertices, so that, for any pair of vertices $u$ and $v$, and any set $F$ of failed vertices, the length of a shortest $u$-to-$v$ path that avoids all vertices in $F$ can be recovered just from the labels of $u$, $v$, and of the vertices in $F$. In this work we study forbidden-set distance labeling schemes in directed planar networks. We also study the extension of such schemes to capture not only the distance from $u$ to $v$, but also the number of distinct $u$-to-$v$ shortest paths.

For unweighted (i.e., unit-weight) graphs, we measure the label size in bits. For weighted graphs and queries concerning lengths
of the shortest paths, we assume that the distance between any two nodes fits in a single machine words, and measure the label size
in words. For queries concerning the number of shortest paths, unless mentioned otherwise, we assume that the number of shortest
paths between any two nodes fits in a single machine word, and measure the label size in words.

\subsection{Related work}

Labeling schemes provide a clean and natural model for studying how to distribute information about a graph. Problems
considered in this model include adjacency~\cite{Kannan,alstrup2015optimal,petersen2015near,alstrup2015adjacency,AlonN17,bonichon2007short},
flows and connectivity~\cite{KatzKKP04,HsuL09,Korman10}, and Steiner tree~\cite{Peleg05}. See~\cite{rotbart2016new} for a recent survey.
We specifically focus on distance labeling schemes.

\paragraph{Distance labeling schemes.}
Embedding distance information into labels was studied by Graham and Pollack~\cite{GrahamP72} in the 1970's in what was termed the squashed cube model. In 2000, Peleg~\cite{Peleg00} formalized the notion of distance labeling schemes, and provided schemes with polylogarithmic label size (number of bits) and query time for trees, interval graphs and permutation graphs. Gavoille et al.~\cite{GPPR04} showed that for general graphs, the label size is $\Theta(n)$, and for trees, $\Theta(\log^2n)$. For (unit-weight) planar graphs they showed a lower bound of $\Omega(n^{1/3})$, and an upper bound of $O(\sqrt n \log n)$ bits. The upper bound was recently improved to $O(\sqrt n)$~\cite{GawrychowskiU16}, but the rare \emph{polynomial} gap between the lower and upper bound remains an interesting and important open problem.
For weighted planar graphs Gavoille et al. gave tight (up to polylogarithmic factors) $\tilde{\Theta}(n^{1/2})$ upper and lower bounds.

\paragraph{Approximate distance labeling schemes.}
Since exact distance labels typically require polynomial size labels~\cite{GPPR04}, researchers have sought smaller labels that yield approximate distances.
Gavoille et al.~\cite{GavoilleKKPP01} studied such labels for general graphs and various graph families.
Specifically, for planar graphs, they presented $O(n^{1/3} \log n)$-bit labels that provide a 3-approximation of the distance.
In the same year, Gupta et al.~\cite{GuptaKR01} presented smaller 3-approximate labels, requiring only $O(\log^2n)$ bits, and Thorup gave $(1+\epsilon)$-approximate labels of size $O(\log n /\epsilon)$, for any fixed $\epsilon>0$~\cite{Thorup04}.
The latter result was generalized to $H$-minor free graphs by Abraham and Gavoille in~\cite{DBLP:conf/podc/AbrahamG06}.

\paragraph{Forbidden-set distance labeling schemes.}
Forbidden-set labels were introduced in the context of routing labels by Feigenbaum et al.~\cite{FeigenbaumKMS05,FeigenbaumKMS07}, and studied by several others~\cite{TwiggPhD,Courcelle09,CourcelleT07,AbrahamCGP16,DBLP:conf/stoc/AbrahamCG12}. 
Exact forbidden-set labeling schemes of polylogarithmic size are given in~\cite{TwiggPhD,CourcelleT07} for graphs of bounded treewidth or cliquewidth. 
For unweighted graphs of bounded doubling dimension, forbidden-set labels with polylogarithmic size and $(1+\epsilon)$-stretch are also known~\cite{AbrahamCGP16}. 
For undirected planar graphs, and for any fixed $\epsilon>0$, Abraham et al.~\cite{DBLP:conf/stoc/AbrahamCG12} presented a forbidden-set labeling scheme of polylogarithmic size such that a
$(1+\epsilon)$-approximation of the shortest path between vertices $u$ and $v$ that avoids a set $F$ of failed vertices can be recovered from the labels of $u,v,$ and the labels of the failed vertices in $\cOtilde(|F|^2)$ time.\footnote{The $\tilde{\cO}(\cdot)$ notation suppresses $\log^{\cO(1)} n$ factors.}

\paragraph{Other related work.}
There are many other concepts related to distances in the presence of failures. 
In the replacement paths problem we are given a graph along with a source and sink vertices, and the goal is to efficiently compute all shortest paths between the source and the destinations for every possible single-edge failure in the graph. In planar graphs this problem can be solved in nearly linear time~\cite{DBLP:journals/talg/EmekPR10,DBLP:journals/talg/KleinMW10,DBLP:conf/soda/Wulff-Nilsen10}.
For the single source, single failure version of the problem (i.e.~when only the source vertex is fixed at construction time, and the query specifies just the target and a single failed vertex), Baswana et al.~\cite{DBLP:conf/soda/BaswanaLM12} presented an oracle with size and construction time $\cO(n \log^4 n)$ that answers queries in $\cO(\log^3 n)$ time.
Building upon this oracle, they then present an oracle of size $\cOtilde(n^2/q)$ supporting arbitrary distance queries subject to a single failure in time $\cOtilde(q)$ for any $q \in [1,n^{1/2}]$.
The authors of~\cite{CharalampopoulosMozesTebeka} show how to construct in $\cOtilde(n)$ time an oracle of size $\cOtilde(n)$ that, given a source vertex $u$, a target vertex $v$, and a set $F$ of $k$ faulty vertices, reports the length of a shortest $u$-to-$v$ path in $G \setminus F$ in  $\cOtilde(\sqrt{k n})$ time.  They further show that for any $r \in [1,n]$ there exists an $\cOtilde(\frac{n^{k+1}}{r^{k+1}} \sqrt{nr})$-size oracle that answers queries in time $\cOtilde(k\sqrt{r})$. Recently, Italiano et al.~\cite{Reachability} gave an oracle of size $\cO(n \log n)$ and construction time  $\cO(n \log^2 n/\log\log n)$ that supports reachability queries subject to a single failure in time $\cO(\log n)$.

Another related concept is that of dynamic distance oracles. Here a graph is preprocessed so as to efficiently support distance queries between arbitrary pairs of vertices as well as updates to the graph. Updates may include deletion of edges or vertices (decremental updates), or also addition of new edges and vertices (fully dynamic).
Fakcharoenphol and Rao~\cite{FR} presented distance oracles that require $\cOtilde(n^{2/3})$ and $\cOtilde(n^{4/5})$  amortized time per update and query for non-negative and arbitrary edge-weight updates respectively.\footnote{Though this is not mentioned in~\cite{FR}, the query time can be made worst case rather than amortized by standard techniques.} The space required by these oracles is $\cO(n \log n)$.
The extensions of this result in~\cite{MSSP,DBLP:conf/stoc/ItalianoNSW11,DBLP:journals/talg/KaplanMNS17,CharalampopoulosMozesTebeka} yield a dynamic oracle that can handle arbitrary edge weight updates, edge deletions and insertions (not violating the planarity of the embedding) and vertex deletions, as well as answer distance queries, in $\cOtilde(n^{2/3})$ time each.

\paragraph{Counting shortest paths.}
In the (non-faulty) counting version of shortest paths labeling, given the labels of vertices $s$ and $t$ we wish to return the {\em number} of shortest $s$-to-$t$ paths in $G$ (i.e.~paths whose length is equal to $d(s,t)$). 
This problem (without faults) was recently studied in~\cite{ISAAC} where labels\footnote{In~\cite{ISAAC}, the authors actually considered the oracle version of the problem, but their solution can be easily applied for labeling as well.} of size $\Theta(\sqrt{n})$ were constructed
under the assumption that the number of shortest paths between any two nodes fits in a constant number of machine words.
In the general case where the numbers consist of $L$ bits, the obtained labels consist of $O(\sqrt{n}\cdot L)$ bits.
As already observed in~\cite{ISAAC}, it is easy to construct an unweighted graph where $L=n-1$ making the labels consist of $\Theta(n^{1.5})$
bits, that is, more than in a naive encoding storing the whole graph in every label. However, the following simple construction shows
that we cannot hope to construct labels consisting of $o(n)$ bits without bounding $L$: given $n$ bits $b_{0},\ldots,b_{n-1}$ we construct
a graph consisting of a path $s=u_{0}-u_{1}-\cdots - u_{n-1}$ and another path $v_{1}-v_{2}-\cdots-v_{n}=t$ in which every edge is duplicated (i.e., there are two parallel edges between each pair $v_i,v_{i+1}$).
Finally, for every $i=0,\ldots,n-1$ such that $b_{i}=1$, we add an edge $u_{i}-v_{i+1}$. Then the number of shortest $s$-to-$t$ paths is exactly $\sum_{i=0}^{n-1} b_{i} \cdot 2^{n-1-i}$, and so by an encoding argument the total number of bits in the labels of $s$ and $t$ must be at least $n$.
Therefore, when counting shortest paths we will measure the size of a label in the number of machine words, each long enough to store the
number of shortest paths between any two nodes in the graph.

We highlight one interesting application where our scheme for counting shortest $s$-to-$t$ paths that avoid nodes $v_{1},v_{2},\ldots,v_{k}$
can be modified to obtain a better bound on the sizes of the labels in bits. Say that instead of counting such shortest paths we would like to
check if avoiding nodes $v_{1},v_{2},\ldots,v_{k}$ increases the length of the shortest path. In such case, we only need to check if the number
of shortest $s$-to-$t$ paths that avoid nodes $v_{1},v_{2},\ldots,v_{k}$ is nonzero. Because the number of shortest paths is always at most
$2^{n}$, by well known properties of prime numbers, choosing a random prime $p$ consisting of $\Theta(k\cdot \log n)$ bits guarantees that with
high probability, for every $s,t,v_{1},v_{2},\ldots,v_{k}$, the number of shortest paths counted modulo $p$ is nonzero if and only if the number
of shortest paths is nonzero. Our scheme (as well as the scheme of~\cite{ISAAC}) can be used for counting modulo $p$, so we obtain
labels consisting of $\tilde{O}(\sqrt{n}\cdot k)$ bits for such queries.

\subsection{Our results}

\begin{itemize}
\item In Section~\ref{sec:exact} we present a single-fault distance labeling scheme (forbidden-set labeling scheme for a set of cardinality 1). The label size is $\cO(n^{2/3})$, the query time is $\cOtilde(\sqrt{n})$, and time to construct all labels is $\cOtilde(n^{5/3})$. Our labeling scheme extends (with no overhead in the label size) to a labeling scheme for counting shortest paths (with a single fault). 

\item In Section~\ref{Labeling for Counting Shortest Paths} we extend the counting labels of~\cite{ISAAC} to the following fault-tolerant variant. Given the labels of vertices $s,t,v_1,v_2,\ldots,v_k$, we wish to return the number of $s$-to-$t$ paths that avoid vertices $v_1,\ldots,v_k$ and whose length is equal to $d(s,t)$ (the original $s$-to-$t$ distance in $G$). 
We show that the labeling of~\cite{ISAAC} (with labels of size $\tilde{O}(\sqrt{n})$) actually works in this more general setting. A naive query to such labeling takes $\tilde{O}(\sqrt{n} \cdot k^2)$ time, we show how to improve this to $\tilde{O}(\sqrt{n} \cdot k)$.

\item In Section~\ref{Labeling Lower Bound for Counting Shortest Paths} we show a lower bound of 
$\Omega(\sqrt{nL})$ on the label-size (in bits) for counting shortest paths (without faults), in graphs in which the number of distinct shortest paths between any two nodes consists of at most $L$ bits.

\item In Section~\ref{A Lower Bound on Dynamic Oracles for Counting Shortest Paths} we show a lower bound on dynamic oracles for counting shortest paths, conditioned on the hardness of online boolean matrix-vector multiplication. 
We prove that for any dynamic shortest paths counting oracle in undirected planar graphs, either the query time or the update time must be $\Omega(\sqrt{n})$ (up to subpolynomial factors).
\end{itemize}

We focus on planar graphs but in fact all our results (except for the efficient preprocessing time and query time in Section~\ref{sec:exact}) hold for any graph family with $\cO(\sqrt{n})$-size separators (such as $H$-minor free graphs and bounded genus graphs). This is also the case for the standard (i.e.~without failures) labeling scheme of Gavoille et al.~\cite{GPPR04}. However, while their $\tilde O(n^{1/2})$-size labels are obtained with a straightforward application of separators, our $\cO(n^{2/3})$-size (fault-tolerant) labels  are obtained with a non-standard and intricate use of separators.    

A main open question that is left unanswered by our work is the existence of non-trivial forbidden-set distance labels tolerating  more than a single fault. 
Labels for approximate distances~\cite{DBLP:conf/stoc/AbrahamCG12} also rely on separators, and do handle multiple failures. 
In the failure-free case,
the labels of~\cite{DBLP:conf/stoc/AbrahamCG12} consist of distances to a small (logarithmic) sample of vertices on some separators, called connections. 
To handle failures, the label of each vertex $u$ also stores the failure-free labels of the connections of $u$. 
This only increases the label-size by a polylogarithmic factor. 
In case of exact distances, the size of the failure-free labels is $\Omega(\sqrt n)$, so this approach seems unsuitable.

Another natural open question is whether the gap between our $\cO(n^{2/3})$-size fault-tolerant labels and the $\tilde O(n^{1/2})$-size labels without failures is actually required and tight. We observe that the existing lower bound technique of Gavoille et al. cannot be extended to show a lower bound above $\Omega (\sqrt n)$ for fault-tolerant labels. The reason is that their technique uses a global argument showing that if we wish to encode the distances between a subset $S$ of $k\le \sqrt n$ vertices then all their labels together require size $\Omega (k^2)$. However, even in the presence of (any number of) failures, encoding distances can be done with total size $\tilde O(k^2)$ (simply store for every $u,v \in S$ the length of the shortest $u$-to-$v$ path that is internally disjoint from $S$).    

\section{Preliminaries}\label{sec:prelim}

Throughout the paper we consider as input a weighted directed planar graph $G$, embedded in the plane.
We assume that the input graph has no negative length cycles.
We can transform the graph in a standard way, in $\cO(n \frac{\log^2 n}{\log \log n})$ time, so that all edge weights are non-negative and distances are preserved~\cite{DBLP:conf/esa/MozesW10}.

\paragraph{Separators and recursive decompositions.}

Miller~\cite{DBLP:conf/stoc/Miller84} showed how to compute a Jordan curve that intersects the graph at a set of nodes $Sep(G)$ of size $\cO(\sqrt{n})$ and separates $G$ into two pieces with at most $2n/3$ vertices each. Jordan curve separators can be used to recursively separate a planar graph until pieces have constant size.
The authors of~\cite{DBLP:conf/stoc/KleinMS13} show how to obtain a complete recursive decomposition tree $\TG$ of $G$ in $\cO(n)$ time. 
$\TG$ is a binary tree whose nodes correspond to subgraphs of $G$ (called {\em pieces}), with the root being all of $G$ and the leaves being pieces of constant size.
We identify each piece $P$ with the node representing it in $\TG$ (we can thus abuse notation and write $P\in \TG$),
 with its boundary $\partial P$ (i.e.~vertices that belong to some separator along the recursive decomposition used to obtain $P$), and with its separator $Sep(P)$.
We denote by $\TG[P,Q]$ the $P$-to-$Q$ path in $\TG$ (and also use $\TG(P,Q]$, $\TG[P,Q)$, and $\TG(P,Q)$). 
 
An \emph{$r$-division}~\cite{DBLP:journals/siamcomp/Frederickson87} of a planar graph, for $r \in [1,n]$, is a decomposition of the graph into $\cO(n/r)$ pieces, each of size $\cO(r)$, such that each piece $P$ has $\cO(\sqrt{r})$ boundary vertices (denoted  $\partial P$).
Another desired property of an $r$-division is that the boundary vertices lie on a constant number of faces  (called holes) of the piece.
For every $r$ larger than some constant, an $r$-division with few holes is represented in the decomposition tree $\TG$ of~\cite{DBLP:conf/stoc/KleinMS13}. It is convenient to describe the $r$-division by truncating $\TG$ at pieces of size $O(r)$, that also satisfy the other required properties. We refer to those pieces (the leaves of $\TG$ after truncation) as {\em regions} and denote by $R_u$ the region containing vertex $u$ (if $u$ belongs to multiple regions, we arbitrarily designate one of them as $R_u$).

\paragraph{Dense distance graphs and FR-Dijkstra.}
The \emph{dense distance graph} of a set of vertices $U$ that lie on a constant number of faces of a planar graph $H$, denoted 
$DDG_H(U)$ is a complete directed graph on the vertices of $U$.
Each edge $(u,v)$ has weight $d_{H}(u,v)$, equal to the length of the shortest $u$-to-$v$ path in $H$.
$DDG_H(U)$ can be computed in time $\cO((|U|^2 + |H|) \log |H|)$ using the multiple source shortest paths (MSSP) algorithm~\cite{MSSP,DBLP:journals/siamcomp/CabelloCE13}.
Thus, computing $DDG_P(\partial P)$ over all pieces of the recursive decomposition of $G$ requires time $\cO(n \log^2 n)$ and space $\cO(n \log n)$.
We next give a --convenient for our purposes-- interface for FR-Dijkstra~\cite{FR}, which is an efficient implementation of Dijkstra's algorithm on any union of $DDG$s.
The algorithm exploits the fact that, due to planarity, certain submatrices of the adjacency matrix of $DDG_H(U)$ satisfy the Monge property.
(A matrix $M$ satisfies the Monge property if, for all $i<i'$ and $j<j'$, $M_{i,j}+M_{i',j'} \leq M_{i',j}+M_{i,j'}$~\cite{monge1781memoire}.) The interface is specified in the following theorem, which was essentially proved in~\cite{FR}, with some additional components and details from~\cite{DBLP:journals/talg/KaplanMNS17,DBLP:conf/esa/MozesW10}.

\begin{theorem}[\cite{FR,DBLP:journals/talg/KaplanMNS17,DBLP:conf/esa/MozesW10}]\label{thm:FR}
Given a set $Y$ of $DDG$s, Dijkstra's algorithm can be run on the union of any subset of $Y$ with $\cO(N)$ vertices in total (with multiplicities) and an arbitrary set of $\cO(N)$ extra edges in time $\cO(N \log^2 N)$.
\end{theorem}

\section{Single-Fault Labeling for Reporting Shortest Paths}\label{sec:exact}

\paragraph{Warm-up.}
As a warm-up, we first sketch a simple labeling scheme that assigns a label of size $\cO(n^{4/5})$ to each vertex.
Consider an $r$-division for $r=n^{4/5}$, and let $\mathcal{R}$ be the set of its regions.
The label of each vertex $u$ consists of the following:
\begin{enumerate}[label=(\alph*)]
\item The $r$-division $\mathcal{R}$. Space: $O(n/r)$.
\item \label{warm2} For each region $R$ in the $r$-division, the length of the shortest path in $G$, among paths that are internally disjoint from $R$, from $u$ to $\bigcup_{P \in \mathcal{R}} \partial P$, and from $\bigcup_{P \in \mathcal{R}} \partial P$ to $u$.
There are $O(n/r)$ regions and for each of them we store $O(n/r \cdot \sqrt{r})$ distances. Space: $O(n^2/r^{3/2})$.
\item \label{warm3} The region $R_u$ and the $\partial R_u$-to-$\partial R_u$ distances in $G\setminus\{u\}$. Space: $O(r)$.
\end{enumerate}
The space is thus $\cO(n/r+n^2/r^{3/2}+r)=O(n^{4/5})$.

Let us now consider a query $(u,v,f)$, and assume, for simplicity, that no two of $u$, $v$ and $f$ are contained in a single region. We have two cases.
If there is a shortest $u$-to-$v$ path in $G \setminus \{ f \}$ that is vertex-disjoint from $R_f$, then the $u$-to-$\partial R_v$ distances among paths internally-disjoint from $R_f$ (item~\ref{warm2}), together with $R_v$, which is stored for $v$ (item~\ref{warm3}), allow us to retrieve the length of this path.
In the other case, we employ the $u$-to-$\partial R_f$ distances among paths internally-disjoint from $R_f$ (item~\ref{warm2}), the information stored in item~\ref{warm3} for $f$, and the $\partial R_f$-to-$v$ distances among paths internally-disjoint from $R_f$ (item~\ref{warm2}).

It is not difficult to combine this approach with the distance-labeling scheme of Gavoille et al.~\cite{GPPR04} for the failure-free setting to obtain labels of size $O(n^{3/4})$.
(Item~\ref{warm2} has to be modified to store distances to separators of ancestors of $R_u$ instead of distances to $\bigcup_{P \in \mathcal{R}} \partial P$, requiring $O(n^{3/2}/r)$ space.)
In the approach that we present below, we rely on separators in a more sophisticated and delicate manner to obtain labels of size $O(n^{2/3})$.

\paragraph{The label.}
Recall that an $r$-division is represented by a decomposition tree $\TG$, whose root corresponds to $G$. The internal nodes of $\TG$ correspond to pieces of $G$. The two children of a piece $P\in \TG$ are the subgraphs of $P$ external and internal to $Sep(P)$. The leaves of $\TG$ are the regions of the $r$-division.

The label of each vertex $u$ in $G$ consists of the following information:
\begin{enumerate}[label=(\roman*)]
\item \label{item1} The entire recursive decomposition tree $\TG$. Space: $O(n/r)$.
\item \label{item2} For each region $R$ in the $r$-division, the shortest path distances in $G$ from $u$ to $\partial R$ among paths that are internally disjoint from $R$. There $O(n/r)$ regions and each of them has $O(\sqrt{r})$ boundary nodes. Space: $O(n/\sqrt{r})$.
\item \label{item3} The region $R_u$ and the $\partial R_u$-to-$\partial R_u$ distances in $G\setminus\{u\}$. Space: $O(r)$.
\item \label{item4} For each piece $P \in \TG$ with sibling $Q$, for each $p\in \partial P \setminus Q$, the shortest path distance from $u$ to $p$ in $G \setminus (P \cup Q) \cup \{p\}$, and the shortest path distance from $p$ to $u$ in $G \setminus Q$. Space: $O(\sum_{P \in \TG} \partial P)=O(n/\sqrt{r})$, c.f.~\cite{vorexact}.
\item \label{item5} For each ancestor piece $P$ of $R_u$ in $\TG$, for each vertex $p$ of $Sep(P) \setminus \partial P$, the shortest path distance from $u$ to $p$ among paths in $P \setminus \partial P$ that are internally disjoint from $Sep(P)$, and the shortest path distance in $P \setminus \partial P$ from $p$ to $u$. Space: $O(\sqrt{n})$, c.f.~\cite{vorexact}.
\end{enumerate}
 
\noindent The overall space required by the above five items is $O(n/r + n/\sqrt{r} + r + n/\sqrt{r} + \sqrt{n})$, which is $O(n^{2/3})$ for $r = n^{2/3}$.

\paragraph{The query.}
Upon query $(u,v,f)$ we say that a path is a $(u,v,f)$-path if it is a $u$-to-$v$ path in $G$ that avoids $f$, and we seek the shortest $(u,v,f)$-path, which we denote by $S$. Let $X$ denote the lowest node in $\TG$ that is an ancestor of $R_f$ and of at least one of $\{R_u,R_v\}$. Let us assume without loss of generality that $X$ is an ancestor of $R_u$. We return the minimum of the following three:    
\begin{enumerate}
\item $S$ includes a vertex of $\partial R_f$. 

The length of this path is found with a SSSP computation on the (non-planar) graph $G_1$ whose vertices are $u,v$, and $\partial R_f\setminus \{f\}$ and whose edges are in one-to-one correspondence with the distances specified below, i.e.~for each $a$-to-$b$ distance, there is an edge from $a$ to $b$ with length equal to that distance:
\begin{itemize}
    \item the $u$-to-$\partial R_f\setminus \{f\}$ distances from item~\ref{item2} in $u$'s label (or the $u$-to-$\partial R_f\setminus \{f\}$ distances in $R_f\setminus \{f\}$, which can be computed from item~\ref{item3}, if $R_u=R_f$);
    \item the $\partial R_f\setminus \{f\}$-to-$\partial R_f\setminus \{f\}$ distances from item~\ref{item3} in $f$'s label;
    \item the $\partial R_f\setminus \{f\}$-to-$v$ distances from item~\ref{item2} in $v$'s label (or the $\partial R_f\setminus \{f\}$-to-$v$ distances in $R_f\setminus \{f\}$, which can be computed from item~\ref{item3}, if $R_v=R_f$).
\end{itemize} 

\item $S$ avoids $R_f$ but includes a boundary vertex of some piece on $\TG[X,R_f]$.

The length of this path is found with a SSSP computation on the graph $G_2$ whose vertices are $u,v$, and $\partial P$ of all nodes $P$ that are siblings of some node $Q$ on the $X$-to-$R_f$ path in $\TG$. 
The edges are in one-to-one correspondence with the $u$-to-$\partial P$ distances from item~\ref{item4} in $u$'s label and the $\partial P$-to-$v$ distances from item~\ref{item4} in $v$'s label.

\item $S$ avoids all boundary vertices of all the pieces on $\TG[X,R_f]$. 

This is required only for the case where the lowest common ancestor of $R_u$ and $R_v$ is not an ancestor of $R_f$ (otherwise,  it is an ancestor of $X$ and a $u$-to-$v$ path cannot avoid the boundary vertices of $X$). 
The length of this path is found with a SSSP computation on the graph $G_3$ whose vertices are $u,v$, and $Sep(P)\setminus \partial P$ of all nodes $P$ on $\TG(X,R_u)$. 
The edges are in one-to-one correspondence with the $u$-to-$Sep(P)$ distances from item~\ref{item5} in $u$'s label and the $Sep(P)$-to-$v$ distances from item~\ref{item5} in $v$'s label. If $R_u=R_v$, the shortest path may not cross any of these separators; in that case the distance may be retrieved by a single SSSP computation in $R_u \setminus \{f\}$ (item~\ref{item3}).
\end{enumerate}

\paragraph{Correctness.}
Let us consider the three options for the shortest $(u,v,f)$-path $S$ (an illustration is provided in Figure~\ref{fig:cases}).
\begin{enumerate}
\item $S$ includes a vertex of $\partial R_f$. 
Let $a$ (resp. $b$) denote the first (resp. last) vertex of $S$ that belongs to $\partial R_f\setminus \{f\}$. The path $S$ can be partitioned into a $u$-to-$a$ prefix, an $a$-to-$b$ infix, and a $b$-to-$v$ suffix. All three subpaths are represented in $G_1$, and all paths represented in $G_1$ do not include $f$. 

\item $S$ avoids $R_f$ but includes a boundary vertex of some piece on $\TG[X,R_f]$. First observe that all $u$-to-$v$ paths in $G_2$ avoid some (not necessarily proper) ancestor of $R_f$ and therefore also avoid $f$. To see that $S$ is represented in $G_2$, let $Q$ denote the unique piece on $\TG(X,R_f]$ such that $S$ avoids $Q$ but visits its sibling $P$ (such a piece $Q$ must exist because $S$ avoids $R_f$ but visits some piece on $\TG[X,R_f]$). Since $S$ visits $P$ it must visit some vertex  of $\partial P$. Let $p$ be the first such vertex of $S$.  Partition $S$ into a shortest $u$-to-$p$ path in $G \setminus (Q\cup P) \cup \{p\}$ and a shortest $p$-to-$v$ path in $G\setminus Q$. These two subpaths are represented in $G_2$.         

\item $S$ avoids all boundary vertices of all the pieces on $\TG[X,R_f]$. 
If $S$ does not visit $\partial R_u$ (and thus $R_u=R_v$) then we find $S$ with an SSSP computation in $R_u \setminus \{f\}$. Otherwise, $S$ visits a separator vertex in of some piece that is a proper ancestor of $R_u$. Let $P$ be the rootmost such piece. Since $S$ avoids $\partial X$ we have that $S$ is restricted to $X$ and hence $P$ is a descendant of $X$. In fact, $P$ must be a {\em proper} descendant of $X$ (otherwise, $S$ visits $Sep(X)$ and therefore visits the boundary of both child-pieces of $X$ including the one on $\TG[X,R_f]$, a contradiction). We therefore have that $P \in \TG(X,R_u)$ and $S$ is restricted to $P$. Also observe that $S$ avoids $\partial P$ because otherwise $S$ must visit a separator vertex of some ancestor of $P$, contradicting $P$ being rootmost. Let $p$ be the first vertex of $S$ that belongs to $Sep(P)$. $S$ can be decomposed into a shortest path from $u$ to $p$ in $P \setminus \partial P$ that is internally disjoint from $Sep(P)$, and a suffix that is a shortest path from $p$ to $v$ in $P \setminus \partial P$; $S$ is thus represented in $G_3$. To see that no path represented in $G_3$ contains $f$, observe that $P$ may contain $f$, but since $R_f$ is not a descendant of $P$, $f$ must be a vertex of $\partial P$ and so is not visited by any path represented in $G_3$.
\end{enumerate}

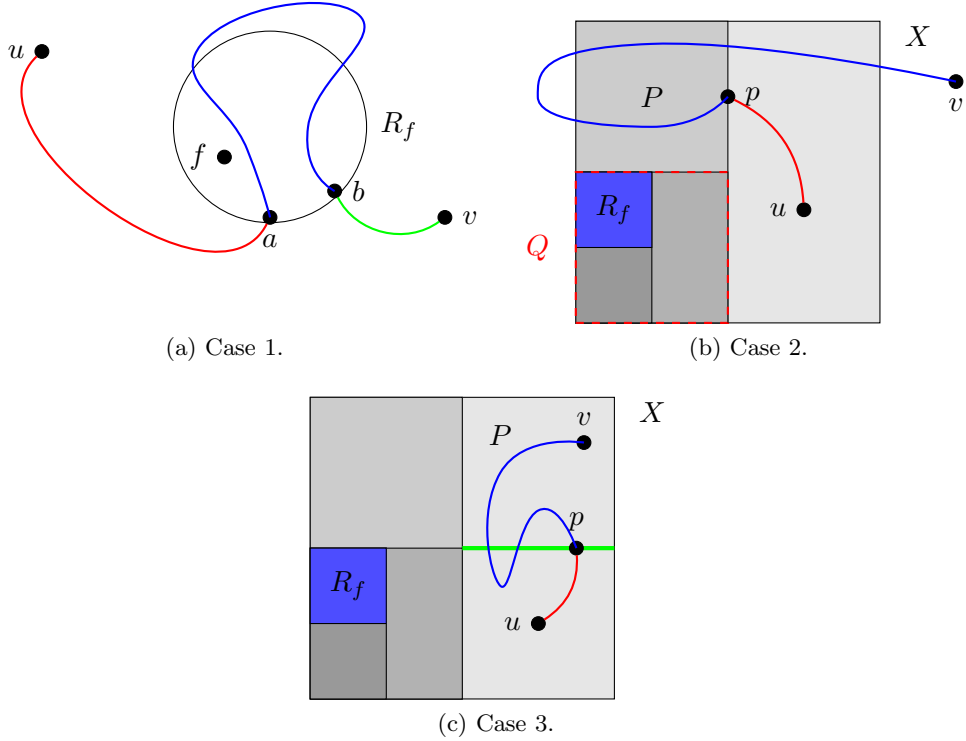
\begin{figure}[t]
\centering
\subfloat[Case 1.]{
{\begin{tikzpicture}
%%%% Nodes %%%%%
\node[shape=circle,scale=7,draw=black, label = {right: $R_f$}] (Rf) at (0,0) {};

\node[shape=circle,scale=0.5,draw=black, fill, label = {left: $u$}] (u) at (-3,1) {};
\node[shape=circle,scale=0.5,draw=black, fill, label = {right: $v$}] (v) at (2.3,-1.2) {};
\node[shape=circle,scale=0.5,draw=black, fill, label = {left: $f$}] (f) at (-0.6,-0.4) {};
\node[shape=circle,scale=0.5,draw=black, fill, label = {below: $a$}] (a) at (0,-1.2) {};
\node[shape=circle,scale=0.5,draw=black, fill, label = {right: $b$}] (b) at (0.85,-0.85) {};

%%%% Paths %%%%%
\path [-,thick, color=red] (u) edge[bend right=100] (a);
\draw[thick,blue] plot [smooth, tension=1] coordinates {(a) (-0.3,-0.3) (-0.9,1.1)  (1.15,1.5) (0.5,0) (b)};
\path [-,thick, color=green] (b) edge[bend right=50] (v);
\end{tikzpicture}}
}
\subfloat[Case 2.]{
{\begin{tikzpicture}
%Shades
\filldraw[draw=black,fill=gray!20]  (0,0)  rectangle (4,4);
\filldraw[draw=black,fill=gray!40]  (0,0)  rectangle (2,4);
\filldraw[draw=black,fill=gray!60]  (0,0)  rectangle (2,2);
\filldraw[draw=black,fill=gray!80]  (0,0)  rectangle (1,2);
\filldraw[draw=black,fill=blue!70]  (0,1)  rectangle (1,2);

%Labels
\draw (0.5,1.5) node {$R_f$};
\draw (4.5,3.8) node {$X$};
\draw (1,3) node {$P$};
\draw[color=red, dashed, thick] (0,0) rectangle (2,2);
\draw[color=red] (-0.5,1) node {$Q$};

%Path
\node[shape=circle,scale=0.5,draw=black, fill, label = {left: $u$}] (u) at (3,1.5) {};
\node[shape=circle,scale=0.5,draw=black, fill, label = {right: $p$}] (p) at (2,3) {};
\node[shape=circle,scale=0.5,draw=black, fill, label = {below: $v$}] (v) at (5,3.2) {};
\path [-,thick, color=red] (u) edge[bend right=30] (p);
\draw[thick,blue] plot [smooth, tension=1] coordinates {(p) (1,2.6) (-0.5,3) (1,3.7) (v)};
\end{tikzpicture}}
}
\newline
\subfloat[Case 3.]{
\begin{tikzpicture}
%Shades
\filldraw[draw=black,fill=gray!20]  (0,0)  rectangle (4,4);
\filldraw[draw=black,fill=gray!40]  (0,0)  rectangle (2,4);
\filldraw[draw=black,fill=gray!60]  (0,0)  rectangle (2,2);
\filldraw[draw=black,fill=gray!80]  (0,0)  rectangle (1,2);
\filldraw[draw=black,fill=blue!70]  (0,1)  rectangle (1,2);

%Separator
%\draw[ultra thick,color=green] (3,2) ellipse (0.7cm and 1.3cm);
\draw[ultra thick,color=green] (2,2) -- (4,2);

%Labels
\draw (0.5,1.5) node {$R_f$};
\draw (4.5,3.8) node {$X$};
\draw (2.5,3.5) node {$P$};

%Path
\node[shape=circle,scale=0.5,draw=black, fill, label = {left: $u$}] (u) at (3,1) {};
\node[shape=circle,scale=0.5,draw=black, fill, label = {above: $p$}] (p) at (3.5,2) {};
\node[shape=circle,scale=0.5,draw=black, fill, label = {above: $v$}] (v) at (3.6,3.4) {};
\path [-,thick, color=red] (u) edge[bend right=30] (p);
\draw[thick,blue] plot [smooth, tension=1] coordinates {(p) (3,2.5) (2.5,1.5) (2.5,3) (v)};
\end{tikzpicture}
}
\caption{An illustration of the 3 different cases that arise for the query. In the figures we assume that $u,v \not\in R_f$ and the different colors in each path represent its decomposition as defined in the proof of correctness. In the figure for Case 2, the blue piece denotes $R_f$, while the siblings of its ancestors in $\TG(X,R_f]$ are denoted by different scales of gray; the deeper the piece is in $\TG$, the darker its color. Piece $Q$ is denoted be the red-dashed rectangle. For Case 3, the setting is the same and in our illustration $P$ is the child of $X$ that is not an ancestor of $R_f$. $Sep(P)$ is denoted by green.}
\label{fig:cases}
\end{figure}

We thus arrive at the following result.

\begin{theorem}\label{thm:dist_labels1}
Given a directed planar 
graph $G$ of size $n$, with real edge-lengths, we can assign an $O(n^{2/3})$-size label to each vertex of $G$ such that upon query $(u,v,x)$, where $u,v,x \in V(G)$, the length of the shortest $u$-to-$v$ path in $G \setminus \{x\}$ can be retrieved from the labels of $u$, $v$ and $x$.
\end{theorem}

\begin{remark}
Let us note, that any graph $G$ of size $n$ from a family of graphs that hereditarily admits $\cO(\sqrt{n})$-size separators (such as $H$-minor free graphs and bounded genus graphs) can be recursively decomposed so that we get an $r$-division (perhaps not with the few-holes property). As our labeling scheme does not require the few-holes property, \cref{thm:dist_labels1} actually applies to any such graph family.
\end{remark}

\paragraph{Extension for counting.}
We now show how to extend our single-fault labeling from reporting $u$-to-$v$ shortest paths in $G \setminus \{f\}$ to  counting the number of $u$-to-$v$ shortest paths in $G \setminus \{f\}$. Our modification does not increase the label size (assuming that each number we store fits into a single word, see the discussion in the introduction).
However, the efficient query algorithm cannot be applied, leading to $\cOtilde(n^{2/3})$ query time. 

In order to extend the labeling scheme for counting, 
for every $u$-to-$v$ shortest path distance which is stored in our label, we also store the number of such $u$-to-$v$ shortest paths. 
The change in query time is that instead of the SSSP computations on $G_1,G_2,G_3$ we use an SSSP computation that counts shortest paths. That is, for each edge in $G_i$ there is a value representing its multiplicity (the value we added to the label), and we want to compute the number of shortest paths with respect to the multiplicities. This extension can be achieved by a trivial extension to Dijkstra's algorithm, resulting in $\cOtilde(n^{2/3})$ query time (In contrast, FR-Dijkstra has no known extension for counting shortest paths). The following lemma proves the correctness of our labeling scheme.

\begin{lemma}
Every shortest path from $u$ to $v$ in $G \setminus \{f\}$ is represented exactly once in the query graphs $G_1,G_2,G_3$.
\end{lemma}

\begin{proof}
The same argument as in the correctness subparagraph proves that every shortest path is represented at least once in the query graphs. It remains to show that every path is represented at most once. Let us consider the three cases for a shortest $(u,v,f)$-path:
\begin{enumerate}
\item $S$ includes a vertex of $\partial R_f$. $S$ is not represented in $G_2,G_3$ because every path that is represented there must avoid an ancestor of $R_f$. $S$ is represented exactly once in $G_1$ because it has a unique decomposition into subpaths $S_1S_2S_3$ where $S_1$ is from $u$ to the first vertex $b_1$ of $S$ in $\partial R_f$, $S_2$ is from $b_1$ to the last vertex $b_2$ of $S$ in $ \partial R_f$, and $S_3$ is from $b_2$ to $v$.

\item $S$ avoids $R_f$ but includes a boundary vertex of some piece on $\TG[X,R_f]$. $S$ is not represented in $G_1$ because all the paths that are represented there touch $R_f$, it is also not represented in $G_3$ since every path there avoids all boundary vertices of all pieces in $\TG[X,R_f]$. To prove that $S$ is represented in $G_2$ exactly once we again show that $S$ can be uniquely decomposed into three subpaths in $G_2$.
Let $P$ be the sibling of some piece $Q \in \TG[X,R_f]$ s.t. $S$ visits $P$, and let $p \in \partial P \cap S$. If $P$ is not the deepest such piece, then $S$ also visits $Q$ but the edge $(p,v)$ in $G_2$ counts only paths in $G \setminus Q$, hence $S$ is not represented as a $u-p-v$ path in $G_2$. If $P$ is the deepest such piece but $p$ is not the first vertex in $\partial P$ that $S$ visits, then the $u$-to-$p$ subpath of $S$ is not represented as an edge $(u,p)$ in $G_2$ since only paths in $G \setminus (P \cup Q) \cup \{p\}$ are. 

\item  $S$ avoids all boundary vertices of all the pieces on $\TG[X,R_f]$. $S$ is not represented in $G_1,G_2$ because every path that is represented there touches some piece in $\TG[X,R_f]$. It is counted exactly once in $G_3$ by a similar argument to case 2 above: $S$ is counted once in $G_3$ by the first separator vertex that $S$ visits in the rootmost piece that it visits. Finally, in the case where $R_u = R_v = R_f$, we perform Dijkstra (with its extension for counting) on $R_u \setminus \partial R_u$.\qedhere
\end{enumerate}
\end{proof}

\paragraph{Efficient queries for planar graphs.}
We can easily achieve $\cOtilde(n^{2/3})$ query-time, since this is the size of the graphs that we construct and can thus perform Dijkstra for SSSP computations. This query time applies to any graph family with $\sqrt{n}$-size separators, such as minor-free graphs. 
On planar graphs, in order to perform queries more efficiently we have to assume random access to the labels of vertices $u$, $v$ and $x$; retrieving them would require $O(n^{2/3})$ time.
We present an $\cOtilde(\sqrt{n})$-time query algorithm for planar graphs at the expense of increasing the labels' size by polylogarithmic factors.

Let us now formally state the main result of~\cite{DBLP:conf/soda/BaswanaLM12}.

\begin{theorem}[\cite{DBLP:conf/soda/BaswanaLM12}]\label{thm:baswana}
Given a weighted directed planar graph $G$ of size $n$ and a source $s \in V(G)$, we can construct in $O(n \log^4 n)$ time an $O(n \log^4 n)$-size data structure, that upon query $(v,x)$, for $v,x \in V(G)$, returns the $s$-to-$v$ distance in $G\setminus \{x\}$ in time $O(\log n)$. 
\end{theorem}

{\emph Cases 2 \& 3.} $G_2$ and $G_3$ are of size $O(\sqrt{n})$ and they can be constructed in $O(\sqrt{n})$ time from the labels of $u,v$ and $f$. We can compute SSSPs in these graphs in $O(\sqrt{n} \log n)$ time using Dijkstra's algorithm.
We handle the subcase of Case 3 in which $R_u=R_v$ and the sought shortest path does not cross $\partial R_u$ as follows.
The label of $u$ additionally stores the single-source single-failure distance oracle of Theorem~\ref{thm:baswana} for graph $R_u \setminus (\partial R_u \setminus \{u\})$ and source $u$. It occupies $\cOtilde(r)=\cOtilde(n^{2/3})$ additional space.
Upon query, we simply query this oracle with $(v,x)$.

{\emph Case 1.} This is the only involved case, as $G_1$ can be of size $\Theta(r)=\Theta(n^{2/3})$ and we aim at performing SSSP computations in time $\cOtilde(\sqrt{n})$.
Let us note that the distances of $u$ to $\partial R_f \setminus \{f\}$ in the case that $R_u = R_f$ can be computed in time $\cOtilde(\sqrt{r})=\cOtilde(n^{1/3})$ if we have stored the oracle of Theorem~\ref{thm:baswana} for graph $R_u$ and source $u$ in the label of $u$. The case $R_v=R_f$ can be treated analogously.

In order to perform efficient SSSP computations we resort to FR-Dijkstra (Theorem~\ref{thm:FR}).
We first make a minor modification to item~\ref{item3} of the label so that the Monge property required for FR-Dijkstra is satisfied: instead of storing $\partial R_u$-to-$\partial R_u$ distances in $G\setminus\{u\}$, we instead store $\partial R_u$-to-$\partial R_u$ distances in $R\setminus\{u\}$ and $\partial R_u$-to-$\partial R_u$ distances in $G\setminus ((R \setminus \partial R)\cup \{u\})$.
This ensures that the set of vertices over which the $DDG$s are built lie on a constant number of faces of the reference graph.
The size of the label is unaffected by this modification. We can then use Theorem~\ref{thm:FR} in a straightforward way to compute the sought shortest path in time $\cOtilde(\sqrt{r})=\cOtilde(n^{1/3})$.

\paragraph{Efficient preprocessing for planar graphs.}
The labels can be naively constructed in $O(n^2)$ time. This is true for any graph family with $\sqrt{n}$-size separators. For the case of planar graphs, we now show that the construction time can be improved to $\cOtilde(n^{5/3})$. 

The complete recursive decomposition of $G$, required for item~\ref{item1}, can be computed in $\cO(n)$ time~\cite{DBLP:conf/stoc/KleinMS13}.
For the rest of the items, we use MSSP data structure for an appropriate subgraph of $G$, or of the reverse graph of $G$, i.e.~$G$ with all its edges reversed.

The multiple-source shortest paths (MSSP) data structure~\cite{MSSP} represents all shortest path trees rooted at the vertices of a single face $g$ in a planar graph.
It can be constructed in $\cO(n \log n)$ time, requires $\cO(n\log n)$ space, and can report any distance between a vertex of $f$ and any other vertex in the graph in $\cO(\log n)$ time. 
Using a simple modification of the underlying graph, presented in~\cite{CharalampopoulosMozesTebeka}, we can ensure that MSSP returns the length of the shortest path that is internally disjoint from a prespecified subset of the vertices of $g$.

To compute the information required for item~\ref{item2} of the labels, we build an MSSP data structure for the reverse graph of $G \setminus (R \setminus \partial R)$ for each piece $R$ in the $r$-division and each of the $O(1)$ holes $g$ on which the vertices of $\partial R$ lie. We then query the sought distances. The time required to construct the MSSP data structures is $\cOtilde(n^2/r)=\cOtilde(n^{4/3})$ and the time required for computing the distances is $\cOtilde(n^2/\sqrt{r})=\cOtilde(n^{5/3})$. The precomputations for items~\ref{item3}, \ref{item4} and the first part of item~\ref{item5} can be done analogously --for item~\ref{item3} we store the distances described in the description of the efficient query implementation.

For the second part of item~\ref{item5}, we can not make use of MSSP, as the shortest path from $u$ to $p \in Sep(P)$ is allowed to cross $Sep(P)$. We can instead build an $\cOtilde(|P|)$-size exact distance oracle for $P \setminus \partial P$ in $\cOtilde(|P|^{3/2})$ time that answers distance queries in $\cOtilde(|P|^\epsilon)$ time, for any constant $\epsilon >0$ (\cite{DBLP:conf/stoc/Charalampopoulos19}); we pick $\epsilon = 1/6$. We then query this oracle for the all distances we need to compute in $P \setminus \partial P$. Over all pieces, the preprocessing time is $\cOtilde(n^{3/2})$ and the sought distances are retrieved in $\cOtilde(n^{3/2} \cdot n^{1/6})=\cOtilde(n^{5/3})$

To wrap up, the global preprocessing time is $\cOtilde(n^{5/3})$ and is upper bounded by the total size of the labels up to polylogarithmic factors.

\section{Labeling for Counting Shortest Paths}\label{Labeling for Counting Shortest Paths}
In this section we design labels such that given the labels of any $k+2$ vertices $s,t,v_1,v_2,\ldots,v_k$, we should return the number of $s$-to-$t$ paths that avoid vertices $v_1, \ldots ,v_k$ and whose length is equal to $d(s,t)$ (the original $s$-to-$t$ distance in $G$). Note that this is the same as returning the number of shortest $s$-to-$t$ paths in $G \setminus \{v_1, v_2, \dots, v_k\}$ only if the length of the shortest $s$-to-$t$ path does not change when $\{v_1, v_2, \dots, v_k\}$ fail. We show that the labeling of~\cite{ISAAC} (with labels of size $O(\sqrt{n})$) actually works in this more general setting and show how to perform a query in $\tilde{O}(\sqrt{n} \cdot k)$ time. 
We assume in this section that edge weights are strictly positive.

\paragraph{The label.}
We first compute a complete recursive decomposition of $G$.
The label of each vertex $v$ in $G$ then consists of the following information:
\begin{enumerate}[label=(\roman*)]

\item For each ancestor piece $P$ of $v$, for every $u \in Sep(P)$,  the number $p_1(v,u)$ and length $d_1(v,u)$ of all $v$-to-$u$ shortest paths in $P\setminus Sep(P) \cup \{ u \}$.
\item   For each ancestor piece $P$ of $v$, for every $u \in Sep(P)$,  the number $p_2(u,v)$ and length $d_2(u,v)$ of all $u$-to-$v$ shortest paths in $P\setminus \partial P$.
\end{enumerate}

In what follows, in the case that $u$ is in many separators of ancestor pieces of $v$, when referring to $d_1(v,u)$, $p_1(v,u)$, $d_2(u,v)$ and $p_2(u,v)$ we mean the values computed for the rootmost such piece.

\paragraph{The query - without faults.}
When there are no faulty vertices, every $s$-to-$t$ shortest path $Q$ in $G$ is uniquely determined by a piece $P$ in the recursive decomposition and a vertex $u\in Sep(P)$.
The piece $P$ is the rootmost ancestor piece of $s$ in the recursive decomposition s.t.~$Q$ visits $Sep(P)$ and therefore does not visit $\partial P$. Such a piece $P$ must be an ancestor of both $s$ and $t$. The vertex $u\in Sep(P)$ is the first vertex of $Sep(P)$ visited by $Q$. $Q$ can thus be decomposed into a prefix $Q_1$ in $P\setminus Sep(P) \cup \{ u \}$ from $s$ to $u$, and a suffix  $Q_2$ in $P\setminus \partial P$ from $u$ to $t$.  
For every possible $u$ we have the number of such $Q_1$ in (i) of $s$ and the number of such $Q_2$ in (ii) of $t$. We therefore add the term $p_1(s,u)\cdot p_2(u,t)$ to the answer. However, we only wish to add this term if $d(s,u)+d(u,t)=d(s,t)$ (otherwise, we are counting non-shortest paths). We have $d(s,u)+d(u,t)$ from the labels of $s$ and $t$. We compute $d(s,t)$ as follows. Let $A[v]$ be the union of separator vertices of all ancestors of $v$. Then 
\begin{equation}\label{eq:d}
d(s,t)  = \min_{u \in A[s] \cap A[t] }  (d_1(s,u) + d_2(u,t)),
\end{equation}

\noindent and the overall query is computed as

\begin{equation}\label{eq:paths}
    paths(s,t)\ \ \  = \!\!\!\!\!\!\!\!\!\!\!\! \sum_{\substack{u \in A[s] \cap A[t]  \text{ s.t }\\ d_1(s,u)+d_2(u,t) = d(s,t)}} \!\!\!\!\!\!\!\!\!\!\!\!p_1(s,u) \cdot p_2(u,t)
\end{equation}

It takes $\tilde{O}(\sqrt{n})$ time to perform such query because there are $O(\sqrt{n})$ vertices in $A[s] \cap A[t]$ and for each of them we perform $\tilde{O}(1)$ calculations. We also compute $d(s,t)$ beforehand in $\tilde{O}(\sqrt{n})$ time.

\paragraph{The query - with faults.}
We begin with an $\tilde{O}( \sqrt{n} \cdot k^2)$ time query and then improve this to $\tilde{O}(\sqrt{n} \cdot k)$. We order the faulty vertices in the increasing order of their distances from $s$ in $G$, and index them $v_1, \ldots ,v_k$ accordingly. For convenience we refer to $s$ as $v_0$ and to $t$ as $v_{k+1}$. Denote by $R[j]$  the number of $s$-to-$v_j$ shortest paths in $G$ that avoid $v_1,\ldots,v_{j-1}$. Denoting by  $paths(v_i,v_j)$  
the number of $v_i$-to-$v_j$ shortest paths in $G$ we obtain the recurrence:
\begin{equation}\label{eq:R old}
R[j] = paths(s,v_j)\ \ \  \!\!\!\! -  \!\!\!\!\!\!\!\!\!\!\!\!\!\!\!\!\!\!\! \sum_{\substack{i < j  \text{ s.t. } \\ d(s,v_i)+d(v_i,v_j)=d(s,v_j)\\}}  \!\!\!\!\!\!\!\!\!\!\!\!\!\!\!\!\!\!\!\! R[i]\cdot  paths(v_i,v_j)
\end{equation}
To see why this recurrence holds, it suffices to show that every shortest path $Q$ in $G$ from $s$ to $v_j$ that visits at least one of $v_1, \ldots ,v_{j-1}$  is counted in the second term exactly once. It is clear that every such path $Q$ is counted at least once, because it can be decomposed into a prefix composed of a shortest path from $s$ to the first $v_i$ that $Q$ visits (i.e.~is counted by $R[i]$) and a suffix composed of a $v_i$-to-$v_j$ path (i.e.~counted by $paths(v_i,v_j)$).
To see why every path $Q$ is counted at most once, notice that every such path $Q$ visits the faulty vertices monotonically with respect to their ordering. In other words, if $Q$ visits some $v_i$ and then some $v_j$ then $i<j$. This holds because if $v_i$ is on a shortest path from $s$ to $v_j$ then $d(s,v_i)<d(s,v_j)$,
and by our ordering of the faulty vertices $i<j$. Since $R[i]$ only counts paths that are internally disjoint from failed vertices, the only time $Q$ is counted is when we count paths of the form  $s\rightsquigarrow v_i \rightsquigarrow v_j$, where $v_i$ is the first faulty vertex $Q$ visits.

Given $R[1],\ldots,R[j-1]$ we can compute $R[j]$ in  $\tilde{O}(\sqrt{n} \cdot j)$ using the recurrence. For each faulty vertex $v_i$ with $i<j$ we perform a $paths(v_i,v_j)$ query as described above which takes $\tilde{O}(\sqrt{n})$ time, so the overall complexity is $\tilde{O}(\sqrt{n} \cdot k^2)$.

\paragraph{Improved query time.}
We now show how to improve the query time from $\tilde{O}(\sqrt{n} \cdot k^2)$ to  $\tilde{O}(\sqrt{n} \cdot k)$. 
In order to achieve this, we cannot afford to compute $paths(v_i,v_j)$ for every pair $i,j$. Instead, we will express $R[j]$ as a summation over $O(\sqrt{n})$ terms that we can compute in $\tilde{O}(1)$ time.

\noindent By combining equations~\eqref{eq:paths} and~\eqref{eq:R old}, and since $paths(s,v_j)$ can be computed in $\tilde{O}(\sqrt{n})$ time, we get that computing $R[j]$ boils down to computing the following double summation:
\begin{equation}\label{eq:comb}
\sum_{\substack{i < j  \text{ s.t. } \\ d(s,v_i)+d(v_i,v_j)=d(s,v_j)\\}}  \!\!\!\!\!\!\!\!\!\!\!\!\!\!\!\!\!\!\!\! R[i]  \ \ \ \ \ \ \sum_{\substack{u \in A[v_i] \cap A[v_j]  \text{ s.t }\\ d_1(v_i,u)+d_2(u,v_j) = d(v_i,v_j)}} \!\!\!\!\!\!\!\!\!\!\!\!\!\!\!\!\!\!\!\!p_1(v_i,u) \cdot p_2(u,v_j)
\end{equation}

\noindent The above sum counts all $s$-to-$v_j$  shortest paths $Q$ that can be decomposed into three parts: \\
$Q_1$ - a shortest $s$-to-$v_i$ path in $G$ (for some $v_i$)  that avoids $v_1, \ldots ,v_{i-1}$.\\
$Q_2$ - a shortest $v_i$-to-$u$ path in $P \setminus Sep(P) \cup \{u\}$ for some $u \in Sep(P)$, where $P$ is defined as the rootmost ancestor of $v_i$ s.t. $Q$ touches $Sep(p)$ ($u$ is the first vertex of $Sep(P)$ in $Q_2$). \\
$Q_3$ - a shortest $u$-to-$v_j$ path in $P \setminus \partial P$.

\noindent We use the same decomposition into $Q_1,Q_2,Q_3$ but sum the terms differently. Denoting $D(s,u) =  \min_i (d(s,v_i)+d_1(v_i,u))$ we compute: 
\begin{equation}\label{eq:R new}
     \sum_{\substack{u \in A[v_j] \text{ s.t. }\\ D(s,u)+d_2(u,v_j) = d(s,v_j)}} \!\!\!\!\!\!\!\!\!\!\!\! \!\!\!\! p_2(u,v_j)  {\color{blue} \sum_{\substack{i < j \text{ s.t. } u \in A[v_i] \text{ and } \\ d(s,v_i)+d_1(v_i,u) = D(s,u)}}\!\!\!\!\!\!\!\!\!\!\!\!\!\!\!\!\! R[i] \cdot p_1(v_i,u)}
\end{equation}

Let us explain equation \eqref{eq:R new}. Denote the inner summation term (in blue) as $F_j(u)$. $F_j(u)$ counts the number of combinations for $Q_1Q_2$ by iterating over every faulty vertex $v_i$ where $i<j$ and $u \in A[v_i]$. For a fixed $v_i$, the number of such combinations is $R[i] \cdot p_1(v_i,u)$. Among all $Q_1Q_2$ combinations, we only want to sum combinations $Q_1Q_2$ that have length $d(s,u)$. Ideally, this could be imposed by adding the condition $d(s,v_i)+d_1(v_i,u) = d(s,u)$ to the inner sum. However, we cannot compute $d(s,u)$ because we do not have the label of $u$. Instead, we add the condition $d(s,v_i)+d_1(v_i,u) = D(s,u)$ where $D(s,u) =  \min_i (d(s,v_i)+d_1(v_i,u))$ (observe that $D(s,u) \ge d(s,u)$). This condition is easy to check using $d_1(v_i,u)$ stored in the label of $v_i$ and the value $d(s,v_i)$ which can be computed beforehand using equation \eqref{eq:d}.
The counting remains correct because in the outer sum we check that $D(s,u)+d_2(u,v_j) = d(s,v_j)$ which only holds if $D(s,u) = d(s,u)$ (because when $D(s,u) > d(s,u)$ then by the triangle inequality we have that $D(s,u)+d_2(u,v_j) > d(s,u)+d_2(u,v_j) \geq d(s,v_j)$). Note that even if $D(s,u) = d(s,u)$ it may be that $D(s,u)+d_2(u,v_j) > d(s,v_j)$. This happens in the case that there are no $s$-to-$v_j$ shortest paths  that visit $u$. In other words, we check that a path $Q_1Q_2Q_3$ is shortest by verifying that $d(s,v_i)+d_1(v_i,u)+d_2(u,v_j) = d(s,v_j)$. This is true iff $D(s,u) = d(s,u)$ and $d_2(u,v_j) = d(u,v_j)$ which means that $Q_1Q_2Q_3$ is indeed a shortest path. 

Observe that in the inner sum we consider only $i<j$. This is because for $i \geq j$ none of the paths from $s$ to $v_j$ that visit $v_i$ is shortest due to the ordering of the faulty vertices.

As for the outer sum, it counts the number of $Q_3$ paths for every $u \in A[v_j]$. Overall, we iterate over every $u \in A[v_j]$ and multiply $F_j(u)$ (the number of $Q_1Q_2$ paths) by $p_2(u,v_j)$ (the number of $Q_3$ paths) and obtain the answer.

Overall, in the $j$'th iteration we compute $R[j]$ using the $F_j(u)$ values according to equation \eqref{eq:R new}. Notice that $F_{j+1}(u)$ is either equal to $F_{j}(u)$ or to $F_{j}(u) + R[j] \cdot p_1(v_j,u)$. We can therefore compute $F_{j+1}(u)$ for every $u \in A[v_j]$ using the just computed $R[j]$ and $F_{j}(u)$. This takes total $\tilde{O}(\sqrt{n})$ time and $\tilde{O}(\sqrt{n}\cdot k)$ time over all the $k+2$ iterations. 

In order to check the distance restrictions in the summations we precompute $d(s,v_j)$ for every $0 \leq j \leq k+1$ and $D(s,u)$ for every $u \in \bigcup_{0 \leq i \leq k+1 } A[v_i]$. The former ($d(s,v_j)$) is computed using \eqref{eq:d}, and the latter ($D(s,u)$) is computed by iterating over every $i$ and $u \in A[v_i]$ and maintaining the minimum value for each $D(s,u)$.
The precomputation of $D(s,u)$ and $d(s,v_j)$ therefore takes $\tilde{O}(\sqrt{n} \cdot k)$ time.

\section{A Lower Bound on Labeling for Counting Shortest Paths} \label{Labeling Lower Bound for Counting Shortest Paths}

In this section we prove the following lower bound on labeling schemes for counting shortest paths (without faults) in graphs
such that the number of distinct shortest paths between any two nodes consists of at most $L$ bits.

The proof is a modification of the approach of Gavoille et al.~\cite{GPPR04} for standard distance labeling.
Their proof proceeds by assigning weights to the edges of a $\sqrt{n}\times \sqrt{n}$ grid graph so
that the shortest path from the $i$-th node in the first column to the $j$-th node in the first row consists of $j-1$ horizontal
edges, followed by $i-1$ vertical edges. Then, the proof hides a single bit in every intersection by creating or not a shortcut.
The shortest paths defined above are still of the same form, up to using the shortcut in case it exists: horizontal edges, possibly a shortcut, and then vertical edges.

\begin{theorem}\label{thm:lb_label}
Any labeling scheme for counting shortest paths in planar graphs such that the number of distinct shortest paths between any two
nodes consists of at most $L$ bits requires labels consisting of $\Omega(\sqrt{nL})$ bits.
\end{theorem}
\begin{proof}
Let us consider a $\sqrt{m} \times \sqrt{m}$ grid graph, weighted as in the proof of Gavoille et al.~from \cite{GPPR04}.
In every intersection, instead of a single $s$-to-$t$ shortcut, we introduce an $\cO(L)$-size gadget -- essentially the one described in the introduction, in our proof that labels of $o(n)$ bits cannot exist if $L$ is unbounded.

More specifically, suppose that we are given $L-1$ bits $b_0,\ldots , b_{L-2}$.
Each edge of the gadget will have weight equal to $1/L$ times the weight of the shortcut in the proof of Gavoille et al.
The gadget consists of a path $s=u_{0}-u_{1}-\cdots - u_{L-1}$ and another path $v_{1}-v_{2}-\cdots-v_{L}=t$ in which every edge is duplicated (i.e., there are two parallel edges between each pair $v_i,v_{i+1}$).
Finally, for every $i=0,\ldots,L-2$ such that $b_{i}=1$, we add an edge $u_{i}-v_{i+1}$.
The number of shortest $s$-to-$t$ paths in the gadget is exactly $\sum_{i=0}^{L-2} b_{i} \cdot 2^{L-1-i}$. Note that this number is congruent to $0$ modulo $2$.
The size of the graph is $n=\Theta(mL)$.

Now, the number of shortest paths from the $i$-th node in the first column to the $j$-th node in the first row is $1$ if all $b_i$'s are equal to $0$ for the gadget at intersection $(i,j)$; otherwise it is equal to $\sum_{i=0}^{L-2} b_{i} \cdot 2^{L-1-i}$.
Hence, each pair $(i,j)$ allows us to recover $\Theta(L)$ distinct bits.
Thus, the labels must consist of $\Omega((\sqrt{m}-1)^2L/(2\sqrt{m}-1))=\Omega(\sqrt{m}L)=\Omega(\sqrt{nL})$ bits.
\end{proof}

We leave the problem of closing the gap between this $\Omega(\sqrt{nL})$ lower bound and the $\cO(\sqrt{n}L)$ upper bound open for further investigation.

\section{A Lower Bound on Dynamic Oracles for Counting Shortest Paths}\label{A Lower Bound on Dynamic Oracles for Counting Shortest Paths}

In this section we consider {\em dynamic oracles} for counting shortest paths (without faults) in undirected planar graphs. That is, data structures that can support queries for counting shortest paths as well as updates to the edge weights. We show a lower bound conditioned on the hardness of Online Boolean Matrix-Vector Multiplication (OMv):

\begin{conjecture}[OMv Conjecture, \cite{OMv}]\label{conjectureOMv}
For every $\epsilon > 0$, there is no $O(N^{3 - \epsilon})$-time algorithm that given an $N \times N$ boolean matrix $M$ and a stream of boolean vectors $v_1, \ldots, v_N$ computes the products $Mv_i$ online (i.e.~computes $Mv_i$ before seeing $v_{i+1}$).
\end{conjecture}

Based on the above conjecture, we prove that for any dynamic shortest paths counting oracle in undirected planar graphs, either the query time or the update time must be $\Omega(\sqrt{n})$ (up to subpolynomial factors). 

\begin{restatable}{theorem}{dlb} \label{dynamicLowerBound}
A dynamic shortest paths counting oracle in undirected $n$-vertex planar graphs with amortized query time $q(n)$ and update time $u(n)$ cannot have $q(n) + u(n) = O(n^{1/2 - \epsilon})$ for any $\epsilon > 0$ unless the OMv conjecture is false. 
This holds even if we only allow edge-weight increments and decrements by 1.
\end{restatable}

\begin{proof}

Our proof follows closely the proof of Abboud-Dahlgaard~\cite{AbboudD16} for dynamic oracles reporting (i.e.~not counting) shortest paths. There are a few subtle differences, but the main difference is that~\cite{AbboudD16} was based on min-plus vector-matrix multiplication while ours is based on standard vector-matrix multiplication. 

\paragraph{Encoding the matrix as a grid.}
We consider a $\sqrt{n} \times \sqrt{n}$ boolean matrix $M$ (i.e.~$N=\sqrt{n}$ in Conjecture~\ref{conjectureOMv}) and encode it using a $(\sqrt{n}+1) \times (\sqrt{n}+1)$ grid $G_M$. For convenience, we index the rows and columns of $M$ as  $1,\ldots,\sqrt{n}$ and the rows and columns of $G_M$ as $0,\ldots,\sqrt{n}$.  The grid $G_M$ contains: 
\begin{enumerate}
\item All horizontal edges of the form $((i,j),(i,j+1))$ except for $((0,j),(0,j+1))$ (i.e.~except for the first grid row). All these edges have the same weight $\sqrt{n}$. 

\item  All vertical edges of the form $((i,j),(i+1,j))$  except for  $((i,\sqrt{n}),(i+1,\sqrt{n}))$ (i.e.~except for the last grid column). The weight of edge $((i,j),(i+1,j))$ is $j+1$. 

\item If $M_{i,j} = 1$ we add an edge $e_{i,j} = ((i-1,j-1),(i,j))$ with weight $\sqrt{n} + j$. 
\end{enumerate}
 
Denote the vertices of the first row $(0,j)$ as $s_j$ and vertices of the last column $(i,\sqrt{n})$ as $t_i$. Consider the shortest $s_j$-to-$t_i$ path. It is easy to see that if $M_{i,j+1} = 0$ then this path is (1) unique, (2) composed of a vertical prefix and and a horizontal suffix, and (3) is of length $\sqrt{n}(\sqrt{n}-j)+i(j+1)$. If however $M_{i,j+1} = 1$ then there are exactly two such shortest paths (one using $e_{i,j}$ and the other using $((i-1,j-1),(i,j-1))$ followed by $((i,j-1),(i,j))$) both of length $\sqrt{n}(\sqrt{n}-j)+i(j+1)$.

\paragraph{The zero matrix grid.}

We would like to make the length of the above shortest paths independent of $i$ and $j$. We  define another $(\sqrt{n}+1) \times (\sqrt{n}+1)$ grid $G_0$ that has no diagonal edges and contains: 
\begin{enumerate}
\item All horizontal edges of the form $((i,j),(i,j+1))$ except for $((0,j),(0,j+1))$ (i.e.~except for the first grid row). All these edges have the same weight $\sqrt{n}$. 

\item  All vertical edges of the form $((i,j),(i+1,j))$ except for  $((i,0),(i+1,0))$ (i.e.~except for the first grid column). The weight of edge $((i,j),(i+1,j))$ is $\sqrt{n}-j+1$. 

\end{enumerate}
 
Denote the vertices of the first column $(i,0)$ of $G_0$ as $s'_i$ and vertices of the first row $(0,j)$ of $G_0$ as $t'_j$. The graph $G$ on which we build the oracle is obtained by connecting the two grids $G_0$ and $G_{M^t}$ (the grid representation of the transpose of $M$). This is done by adding edges $b_i = (t_i,s'_i)$ of weight $w(b_i) = (\sqrt{n}+1)(\sqrt{n}-i)$ for every $1 \le i \le \sqrt n$. 

\paragraph{The reduction.}
In order to solve the OMv problem, for each query vector $v[1,\ldots,\sqrt n]$, if $v$ is the all-zero vector we simply output an all-zero vector. Otherwise, we (1) reset the weight $w(b_i)$ of every $b_i$ to be $(\sqrt{n}+1)(\sqrt{n}-i)$, (2) for every $i$, if $v[i] = 0$ we increase the weight of $b_i$ by 1, and (3) for every index $0 \leq j < \sqrt{n}$ we query the oracle for the number of shortest paths from $s_j$ to $t'_{j+1}$. Finally, we decrease the oracle's answer by the number of 1's in $v$ and assign this value as the $j$'th entry in the result $Mv$.

To see why the above procedure correctly calculates $Mv$, first note that the $j$'th entry in $Mv$ is exactly the number of indices $1 \leq i \leq \sqrt{n}$ s.t.~$e_{i,j}$ is present in $G_{M^t}$ and $v[i] = 1$.
The length of the shortest path from $s_j$ to $t'_{j+1}$ through an edge $b_k$ is $d(s_j,t_k) + w(b_k) + d(s'_k,t'_{j+1}) = 2n + 2\sqrt{n}$. This value is independent of both $j$ and $k$, so for each $1  \leq k \leq \sqrt{n} $ we have a unique shortest path through $b_k$ if $e_{k,j + 1}$ is absent in $G$ or exactly two shortest paths if $e_{k,j + 1}$ is present in $G$. 
In step (2), when we increase by 1 the edges $b_i$ corresponding to entries in $v$ where $v[i] = 0$, paths going through these $b_i$'s are longer than paths going through other $b_i$'s and are therefore not shortest (we made sure that $v$ is not all-zero). Hence, every $b_i$ that corresponds to $v[i] = 1$ contributes $1 + M^t_{i ,j + 1}$ to the number of shortest paths from $s_j$ to $t'_{j+1}$, and by subtracting the number of $1$'s in $v$ we obtain the correct answer. 

Overall, for each vector we perform $O(\sqrt{n})$ updates and queries, so overall we perform $O(n)$ updates and queries. If each update/query takes $O(n^{0.5-\epsilon})$ time then we get overall $O(n^{1.5-\epsilon}) = O(N^{3-\epsilon/2})$ contradicting Conjecture~\ref{conjectureOMv}.
\end{proof}

\bibliographystyle{plainurl}

\end{document}